\documentclass[runningheads]{llncs}

 \def\publish{}

\usepackage{cite}
\usepackage{amsmath}
\usepackage{amssymb}
\usepackage{amsfonts}
\usepackage{algorithm}
\usepackage[noend]{algpseudocode}
\usepackage{graphicx}
\usepackage{textcomp}
\usepackage{xcolor}
\usepackage{xspace}
\usepackage{cleveref}
\usepackage{pifont}

\newcommand{\para}[1]{\vskip 1em\noindent\textbf{#1.}}
\newcommand{\aterm}{\langle T \rangle}

\newcommand{\sysname}{Cuttlefish\xspace}
\newcommand{\sys}{\sysname}
\newcommand{\fun}{FastUnlock\xspace}

\newcommand{\none}{\text{\textbf{None}}\xspace}
\newcommand{\unlocked}{\text{\textbf{Unlocked}}\xspace}
\newcommand{\confirmed}{\text{\textbf{Confirmed}}\xspace}
\newcommand{\noop}{\text{\textbf{No-Op}}\xspace}

\newcommand{\Cert}{\textsf{Cert}\xspace}
\newcommand{\Oid}{\textsf{ObjectId}\xspace}
\newcommand{\Okey}{\textsf{ObjectKey}\xspace}
\newcommand{\Okeys}{\textsf{ObjectKeys}\xspace}
\newcommand{\Version}{\textsf{Version}\xspace}

\newcommand{\tx}{\textsf{Tx}\xspace}
\newcommand{\cert}{\textsf{Cert}\xspace}
\newcommand{\auth}{\textsf{Auth}\xspace}
\newcommand{\unlockreq}{\textsf{UnlockRqt}\xspace}
\newcommand{\unlockvote}{\textsf{UnlockVote}\xspace}
\newcommand{\unlockcert}{\textsf{UnlockCert}\xspace}
\newcommand{\effectvote}{\textsf{EffectSign}\xspace}
\newcommand{\effectcert}{\textsf{EffectCert}\xspace}

\newcommand{\valid}[1]{\emph{valid}(#1)\xspace}

\newcommand{\sign}[1]{\emph{sign}(#1)\xspace}

\newcommand{\exec}[1]{\emph{exec}(#1)\xspace}

\newcommand{\lockdb}{\textsc{LockDb}\xspace}
\newcommand{\unlockdb}{\textsc{UnlockDb}\xspace}

\newcommand{\one}{\ding{202}\xspace}
\newcommand{\two}{\ding{203}\xspace}
\newcommand{\three}{\ding{204}\xspace}
\newcommand{\four}{\ding{205}\xspace}
\newcommand{\five}{\ding{206}\xspace}
\newcommand{\six}{\ding{207}\xspace}
\newcommand{\seven}{\ding{208}\xspace}
\newcommand{\eight}{\ding{209}\xspace}

\begin{document}

\title{
    \sys: Expressive Fast Path Blockchains with \fun
}

\ifdefined\publish
    \author{
        Lefteris Kokoris-Kogias\inst{1,2} \and
        Alberto Sonnino\inst{1,3} \and 
        George Danezis\inst{1,3}
    }
    \institute{
        MystenLabs \and
        IST Austria \and 
        University College London
    }
\else
    \author{}
    \institute{}
\fi
\maketitle

\begin{abstract}
    \sys addresses several limitations of existing consensus-less and consensus-minimized decentralized ledgers, including restricted programmability and the risk of deadlocked assets. The key insight of \sys is that consensus in blockchains is necessary due to contention, rather than multiple owners of an asset as suggested by prior work. Previous proposals proactively use consensus to prevent contention from blocking assets, taking a pessimistic approach. In contrast, \sys introduces collective objects and multi-owner transactions that can offer most of the functionality of classic blockchains when objects transacted on are not under contention. Additionally, in case of contention, \sys proposes a novel `Unlock' protocol that significantly reduces the latency of unblocking contented objects. By leveraging these features, \sys implements consensus-less protocols for a broader range of transactions, including asset swaps and multi-signature transactions, which were previously believed to require consensus.
\end{abstract}

\section{Introduction}

Consensus is not required for implementing decentralized asset transfers~\cite{guerraoui19consensus}. This insight led to the design of cryptocurrencies based on consistent or reliable broadcast~\cite{fastpay,zef,astro}, which offer several advantages. They exhibit exceptionally low latency, operate purely asynchronously, and are highly scalable.

However, consensus-less systems suffer from two significant limitations. Firstly, they have limited programmability, 
since to maintain liveness 
transactions must be submitted in a valid and race-condition-free manner. Failure to do so can result in deadlocked assets that become forever inaccessible to their owners. Thus, programmability is restricted to simple transactions involving objects owned by a single entity, such as asset transfers or payments. Attempts to support more complex transactions involving multiple users (e.g., asset swaps) or authorization (e.g., multi-signature) risk causing deadlocks, rendering assets unusable indefinitely. As a result, existing consensus-less cryptocurrencies~\cite{fastpay,zef,astro} are only suited for basic operations.

The second limitation arises from the strong requirement imposed on clients in consensus-less systems to never issue conflicting transactions. Even minor bugs in client implementations can lead to deadlocked assets. For instance, a faulty wallet that unintentionally sends a transaction with a randomized signature twice may be interpreted as two conflicting transactions, resulting in locked assets. Similar issues occur when a client underestimates the required gas for a transaction and attempts to reissue it with a higher gas amount.

To address these limitations, the recent Sui blockchain~\cite{suiL}, introduces a hybrid model that combines a consensus-less \emph{fast path} with a consensus-based fallback. Sui supports general-purpose smart contracts by utilizing broadcast in the fast path, while transactions involving shared objects at risk of race conditions are processed through the consensus fallback path. Consensus is also employed for a daily system reconfiguration which drops all the locks of the fast path and thus implicitly solves potential deadlocks. While Sui theoretically overcomes both limitations, it does so at the cost of increased latency and reduced usability. Accessing shared state necessitates sequencing transactions through consensus in all scenarios, resulting in higher latency ranging from seconds instead of fractions of a second, and prohibiting the use of parallel broadcast protocols for scalability. Furthermore, deadlocked objects are only reset and made available once a day upon reconfiguration, which proves impractical for objects that may legitimately be used by multiple non-coordinating users and get accidentally locked. The fear of losing access to resources for an entire day terrifies developers as it happens on a daily basis due to honest mistakes. As a result, developers on Sui forfeit the use of the low-latency path of Sui and fallback on implementing their smart contract using consensus~\cite{deadlock-fear}.

This paper introduces \sys, a solution that addresses both challenges mentioned above. Extending the high-level design of Sui, \sys combines a fast path reliable broadcast with a consensus fallback while enabling the execution of consensus-less transactions, without introducing additional latency unless there is actual contention.
The first contribution of \sys is the enhancement of fast path object programmability. It achieves this by introducing collective objects, which allow for complex access control and enable the execution of any type of transaction on the fast path. The transaction model is extended to support multi-owner transactions on the fast path, thereby expanding the range of programmable actions beyond single-owner operations.
The second and core contribution of \sys is the design of \fun, a mechanism that facilitates the concurrent execution of conflicting transactions on the fast path, ensuring liveness and enabling rapid unlocking of fast path objects. This advancement makes \sys suitable for a broader spectrum of transactions, including asset swaps and multi-signature transactions.
\sys is currently considered for adoption by the Sui blockchain team. 

\section{Motivating Applications of \sys}\label{sec:applications}
\sys addresses real-world needs raised by existing blockchains.

\para{Deadlocks due to multiple or buggy clients}
\sysname tackles the common challenge of locked objects in consensus-less blockchains, that leads to a poor user experience. Using multiple wallets for the same account or objects can result in concurrent conflicting transactions due to bugs, lack of synchronization, or being offline. Even, a single wallet may send a transaction with insufficient gas, only to later attempt to rectify the mistake by updating the gas value, and leading to two conflicting transaction on the same account or objects. Unfortunately, these innocent slip-ups or bugs are interpreted as equivocation attempts within the context of consistent broadcasts, potentially deadlocking the assets involved. Unlike previous solutions, \sysname enabled users to swiftly regain control of their assets through \fun and retry their transactions safely.

\para{Atomic swaps}
Atomic swaps allow two parties to exchange digital assets without the need for a trusted intermediary. While consensus-based blockchains can achieve this through smart contracts, the risk of deadlock arises in consensus-less environments due to the possibility of a Byzantine user issuing a concurrent transaction. Such a situation would effectively deadlock the assets of both parties. However, this risk only materializes when an active attacker intentionally causes contention. In rare cases like these, the \fun protocol enables participants in the swap to quickly recover their assets. This safety net allows \sysname to support multi-owner transactions in the fast path, allowing fast path atomic swaps and other multi party smart contracts, enhancing the programmability of consensus-less transactions.

\para{Regulated stablecoins}
Regulated stablecoins, require the issuer to be able to block accounts or balances for regulatory reasons, besides their owner spending them, which eludes consensus-less systems. Since multiple parties need to operate on such objects they need to use consensus to sequence these potentially conflicting operations, even though the issue nearly never execises their ability to block objects (creating no practical contention). \sysname allows for collective objects, that may be used by more than one owner, or any pattern or complex access control, and can be used in the fast path.

\section{Background} \label{sec:background}

A number of consensus-less systems have been proposed in the literature, including FastPay~\cite{fastpay}, Astro~\cite{astro}, Zef~\cite{zef}, and Linera~\cite{linera}. We will specifically describe
and extend Sui (the Sui Lutris mechanism~\cite{sui}) as a basis for the \sysname design, as it is the only currently deployed mechanism with a consensus-less fast path. \sysname extends the expresivity of both object authentication and transactions in the Sui fast path, and also extends that Sui consensus path to support \fun.

\para{Object Types} All Sui blockchain state is composed on a set of objects. There are three types of objects, and their use in a transaction determines whether the fast path or the consensus path is to be used.
\begin{itemize}
    \item \emph{Read-only objects} cannot be mutated or deleted and may be used in any type of transactions concurrently and by all users.
    \item \emph{Owned objects} have an owner field that determines access control. When owner is an address representing a public key, a transaction may access the object, if it is signed by that address (which can also be a multi-signature).
          When the owner of an object (called a child object) is another object ID (called the parent object), the child object may only be used if the root object (the first one in a tree of possibly many parents) is included in the transaction and authorized. This facility is used to construct efficient collections.
    \item \emph{Shared objects} do not specify an owner. They can instead be included in transactions by anyone, and do not require any authorization. Instead, they perform their authorization logic (enforced by the smart contract). 
\end{itemize}

\para{Transactions} A transaction is a signed command that specifies several input objects, a version number per object, and a set of parameters. If valid it consumes the input object versions and constructs a set of output objects at a fresh version---which can be the same objects at a later version or new objects. Owned objects versions need to be the latest versions in validator databases, and not be re-used across transactions. Shared objects need not specify a version, and the version on which the transaction is executed is assigned by the system. A transaction is signed by a single address and therefore can use one or more objects owned by that address. A single transaction cannot use objects owned by more than one address.

\para{Certificates}
A \emph{certificate} ($\cert$) on a transaction contains the transaction itself as well as the identifiers and signatures from a quorum of at least $2f+1$ validators.
%
A certificate may not be unique, and the same logical certificate may be signed by a different quorum of validators. However, two different valid certificates on the same transaction should be treated as representing semantically the same certificate. The identifiers of signers are included in the certificate (i.e., accountable signatures~\cite{boneh2018compact}) to identify validators ready to process the certificate, or that can serve past information required to process the certificate.


\begin{figure}[t]
    \centering
    \includegraphics[width=1\columnwidth]{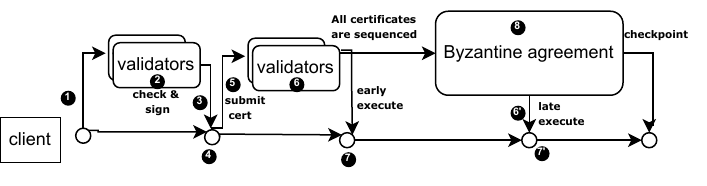}
    \caption{
        General protocol flow of Sui Lutris~\cite{suiL} fast-path \& consensus failover system.
    }
    \label{fig:overview}
\end{figure}

\para{Processing in the Fast Path and Consensus}
\Cref{fig:overview} provide an overview of Sui-Lutris and by extension Cuttlefish's common-case. A transaction is sent  by a user to all validators~(\one), that ensure it is correctly signed for all owned objects and versions, and also that all objects exist~(\two); a correct validator rejects any conflicting transaction using the same owned object versions, in the same epoch (so the first transaction using an object acquires a \emph{lock} on it). They then countersign it~(\three) and returns the signature to the user. A quorum of signatures constitutes a \emph{certificate} for the transaction~(\four). Anyone may submit the certificate to the validators~(\five) that check it.

At this point execution may take the fast path: if the certificate only references read-only and owned objects it is executed immediately~(\six) and a signature on the effects of the execution returned to the user to create an effects certificate~(\seven) and the transaction is final. If any shared objects are included execution must wait. In all cases, certificates are input into consensus and sequenced~(\eight). Once sequenced, the system assigns a common version number to shared objects for each certificate, and execution can resume (steps~\six' and~\seven') to finalize the transaction. The common sequence of certificates is also used to construct checkpoints, which are guaranteed to include all finalized transactions~(\eight).

\para{Checkpoints and Reconfiguration} \label{sec:sui-checkpoints} Sui ensures transaction finality either before consensus for owned object transactions~(\seven) and after consensus for shared object transactions~(\seven'). Its reconfiguration protocols ensure that if a transaction could have been finalized it will eventually be included in a checkpoint before the end of the epoch. At the end of the epoch, all locks are reset~(\two).
%
%
Appendix~\ref{sec:reconfiguration} summarizes the reconfiguration protocol of Sui that \sys directly adopts.

\para{Limitations of Sui}
Misconfigured clients may create and submit concurrently conflicting transactions in step~(\one) and~(\two), that reuse the same owned object versions. In that case, neither transaction may be able to construct a certificate~(\four), and the owned object becomes locked until the end of the epoch. Due to the risk that owned objects can become locked through conflicting transactions, Sui restricts transactions to only contain objects from a single owner, thus limiting the applicability of the fast path---to avoid mistrusting users from locking each others' objects for a day. For similar reasons, objects may also have at most one owner. \sysname addresses all these limitations.

\section{Overview} \label{sec:overview}
\sysname adopts the high-level design of Sui, namely using reliable broadcast for a fast path with a fall-back to consensus (see \Cref{sec:primitives} for definitions of distributed systems primitives), but augments it in the following ways.
\begin{enumerate}
    \item It allows for \emph{multi-owner transactions} on the fast path that use objects with different `owners'. In Sui this can only be expressed with shared objects in transactions using the higher-latency consensus path.
    \item It introduces \emph{collective objects} that allow for complex authorization involving different users or combinations of users, or even time or external events. Collective objects extend owned objects and may be used on the fast path.
    \item It adds a \emph{\fun protocol} that allows for fast path objects blocked due to concurrent conflicting transactions to recover liveness within seconds, whereas Sui would recover only within a day.
\end{enumerate}

Collective objects and multi-owner transactions allow for more expressive transactions in the fast path, but risk increasing the incidence of conflicting transactions and locked owned objects.
To alleviate this issue, \Cref{sec:unlock} presents a simple design for \fun: it performs a no-op on locked objects using the consensus path making them available again within seconds. \Cref{sec:multi-unlock} extends the \fun protocol to force a specific transaction instead of a no-op, which is necessary when objects are under continuous contention.

\fun leverages the consensus protocol to signal that an owned object is suspected of being under contention and should not be processed by the fast path. Following invocations, the current version of the object is blocked, and consensus is used to determine whether a transaction on it might have been final; if not, in the simple \fun the version of the object is increased with a no-op. As a result, the object has a new version that can be accessed via the fast path once again. Since the version is always updated, the transactions that blocked the object are no longer valid, removing any replay attack opportunity. This is not true in Sui as even after the end of the epoch a malicious client can resubmit the equivocated transactions and try to re-lock the object.

\para{Threat Model}
\sysname operates in the same threat model as Sui.
It assumes a message-passing system with a set of $n$ validators and a computationally bound adversary that controls the network and can corrupt up to $f < n/3$ validators within any epoch. We say that validators corrupted by the adversary are \emph{Byzantine} or \emph{faulty} and the rest are \emph{honest} or \emph{correct}.
To capture real-world networks we assume asynchronous \emph{eventually reliable} communication links among honest validators. That is, there is no bound on message delays and there is a finite but unknown number of messages that can be lost.
Similarly to Sui~\cite{sui}, \sysname additionally uses a consensus protocol as a black box that takes some valid inputs and outputs a total ordering~\cite{bullshark,danezis22narwhal,gelashvili22jolteon}, possibly operating within a partially synchronous model~\cite{dwork1988consensus}.

\section{Enhancing Programmability}
\sys provides greater objects programmability on the fast path than existing consensus-less systems using two main ingredients: (i) multi-owner transactions, and (ii) collective objects.

\para{Multi-owner transactions} Sui Lutris requires all owned objects in a transaction to be `owned' by the same address~\cite{sui}. \sys lifts this restriction: a transaction can reference owned objects with any root authenticator term. The transaction contains the authentication evidence used to authorize all objects, such as a set of signatures over the transaction, potentially from multiple addresses.

Validators must ensure that all owned objects referenced by a transactions are correctly authorized before signing a transaction, which ensures that a valid certificate represents an authorized transaction. Transactions that only contain owned objects, even when they have different owners, can be executed on the fast path. Then in addition they contain shared objects their execution needs to be deferred after the certificate has been sequenced by consensus.

Multi-owner transactions make \sysname more susceptible to owned-objects being locked through error or malicious behaviour. For example, consider an atomic swap T transaction that takes objects A owned by Alice and object B owned by Bob and exchanges their ownership. If Alice signs T first, Bob may refuse to sign initially denying Alice access to her object. If Alice loses patience and tries to use A in another transaction T', then Bob can sign T and race Alice's attempt to build a certificate. Now both T and T' contain A and conflict which can lead to a locked A (and B). To resolve such situations it is necessary for \sysname to implement \fun described in \Cref{sec:unlock}.

\para{Collective Objects} Collective objects are owned objects with a more complex authenticator, than the usual address or object ID that Sui Lutris supports. Complex authenticators allow conjunction, disjunction, and weighted thresholds thresholds of authentication terms to be used as authenticators. Authentication terms include the traditional address and object ID, but also conditions on time or events that have occured in the environment of the execution. Due to the fact that multiple non-coordinating or even mutually distrustful parties can use the object in transaction, as well as the fact that some authorization terms are non-deterministic, complex authenticators can lead to conflicting transactions being authorized on objects and thus require the \fun protocols to be practical.

More specifically \sys extends the authorization logic of an owned object to be a root authentication term $\aterm$ from the grammar in \Cref{fig:authenticators}:
\begin{figure}
    \centering
    \begin{align*}
        \aterm & := \textsc{PublicKey}(\mathit{pk})\, | \, \textsc{ObjectID}(\mathit{oid}) \\
               & := \textsc{BeforeTime}(t)\, |\, \textsc{AfterTime}(t)                     \\
               & := \textsc{EventOccured}(c, e)                                            \\
               & := \textsc{Threshold}(W, [(w_i, \aterm_i)])                               \\
               & := \bigwedge\limits_{i} \aterm_i \qquad (\textsc{And})                    \\
               & := \bigvee\limits_{i} \aterm_i \qquad (\textsc{Or})
    \end{align*}
    \caption{Grammar defining the authrorization logic for collective objects}
    \label{fig:authenticators}
\end{figure}

\begin{itemize}
    \item A $\textsc{PublicKey}$ term is true if the transaction is signed by the public key $\mathit{pk}$. Using a single such term as an authenticator for an object expresses the authentication logic of a traditional single owner object in Sui.
    \item A $\textsc{ObjectID}$ term requires the object with id $\mathit{oid}$ to be included (and authenticated) as part of the transaction. A single object id authenticator expresses the traditional parent-child relation, and ownership rules in Sui.
    \item The $\textsc{BeforeTime}$ and $\textsc{AfterTime}$ are true if the (local) time the transaction is received by the validator is respectively before or after  $t$. Note that since even honest validators cannot have perfectly synchronized clocks, it is possible that a transaction with such a term becomes `stuck'.
    \item The $\textsc{EventOccured}$ term becomes true if in the trace of finalized executions a specific event was emitted on chain $c$. Note that the chain may be different chain than the one operated by \sysname effectively making authorization conditional on an oracle for another chain. Such an event may be described by type or content and we abstract this in $e$. A reference to the transaction that emitted the event can be provided as an authenticator to help validators check this term.
    \item The $\textsc{Threshold}$ defined a threshold $W$ and a weight $w_i$ for a set of terms. It is true if the sum of weights of the true terms exceed the threshold. It allows the definition of flexible policies such as requiring a threshold of signature or other conditions to be present to authorize the object being used.
    \item The $\textsc{And}$ and $\textsc{Or}$ define a number of terms, and are true if all or any of these terms are true, respectively.
\end{itemize}
A transaction needs to provide evidence that all authenticator terms for all objects in its input set are true. For each input object it specifies the path(s) in the authentication term tree that are true supporting the overall authenticator term, collectively called \emph{authentication paths}. It also contains a set of signatures (as a list ordered by public key) signing the transaction. To allow for greater flexibility the authentication paths are not signed (conceptually they are part of the signature not the transaction), and therefore a transaction cannot get information about the logic that authorized its execution through this mechanism.

We note that an authorization path may be expressed in a very succinct manner as a one bit per $\textsc{Threshold}$, $\textsc{And}$ or $\textsc{Or}$ branch pursued to demonstrated the root authenticator term to be true. A single signature is required to satisfy any number of $\textsc{PublicKey}$ terms with the same $\mathit{pk}$. $\textsc{ObjectID}$ terms can be demonstrated as satisfied implicitly by including the $\mathit{oid}$ as an input. $\textsc{EventOccured}$, $\textsc{BeforeTime}$ and $\textsc{AfterTime}$ terms are satisfied (or not) through the validator comparing their specified time with the current time or consulting a chain for an event, and incur no additional overhead in terms of evidence in the transaction.

We represent the authenticator logic as a tree, with AND/OR, k-out-of-n connectives as branches and identities, time conditions and object IDs as leafs. In this representation, we can augment each branch and leaf with an optional nonce, compute a Merkle tree over them, and only store a hash of the root as the authenticator. In this way transfering to a complex authenticator is no different than transferring an object to an address, and one cannot tell the difference until the object is used in a transaction. A transaction then reveals only the paths necessary to show that the condition for access is satisfied. This allows objects to preserve secret authenticators until they are accessed, and even upon access only reveal the information required. We leave using zero-knowledge proofs as evidence all authenticators are satisfied in a transaction, allowing us to hide all information besides authorization, for future work.

\section{Baseline \fun Protocol} \label{sec:unlock}

Both multi-owner transactions and collective objects can result in deadlocks in the fast path when correct clients attempt to access objects concurrently. To remedy this issue \sysname introduces a \fun functionality. For simplicity, we show how to unlock a single object by either executing a preexisting transaction to finality or executing a no-op which only increases the version.
\Cref{sec:multi-unlock} extends the basic protocol to execute a new transaction instead of a no-op. Sui~\cite{sui} provides detailed specifications and implementations of its system model and \sysname largely extends it with the additional \fun protocol.


\para{New Persistent Data Structures}
Each \sysname validator maintains a set of persistent tables abstracted as key-value maps, with the usual $\mathsf{contains}$, $\mathsf{get}$, and $\mathsf{set}$ operations.
The map
$$\lockdb[\Okey] \rightarrow \Cert \text{ or } \none$$
maps each object's identifier and version, $\Okey = (\Oid, \Version)$, to a certificate $\cert$ or $\none$ if the object's version exist by the validator does not hold any certificate.
The  map
$$\unlockdb[\Okey] \rightarrow \unlocked \text{, } \confirmed \text{, or } \none$$
records whether a transaction over the specified object version is involved in a current \fun instance ($\unlocked$), has been sequenced by the consensus engine ($\confirmed$), or none of the above ($\none$).

All new owned object entries start with $\unlockdb[\Okey]$ set to \none. Once a transaction certificate is sequenced through consensus it is always executed (whether it is for a shared object transaction or an owned object only transaction) and all owned object entries have $\unlockdb[\Okey]$ set to \confirmed.

\para{\fun Protocol Description} \label{sec:unlock-protocol}
In order to safely unlock an object, the user interactively constructs a proof, called a \textit{no-commit certificate}, that no transaction modifying that object has been committed or will be committed on the fast path. This proof consists of a message signed by a quorum of validators attesting that they have not already executed a transaction over the $\Okey$, and promising that they will not execute any transaction over the $\Okey$ in the fast path. Only certificates sequenced over consensus may affect such an $\Okey$ going forward.

\begin{figure*}[t]
    \centering
    \includegraphics[width=\textwidth]{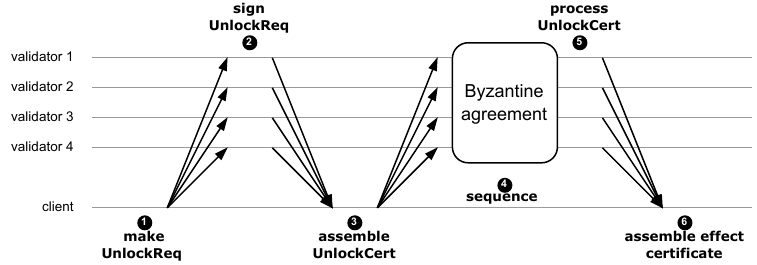}
    \caption{
        \fun interactions between a user and validators to unlock an object.
    }
    \label{fig:fast-unlock}
\end{figure*}

\Cref{fig:fast-unlock} illustrates the fast-unlock protocol allowing a user to instruct validators to unlock a specific object.
A user first creates an \emph{unlock request} specifying the object they wish to unlock:
$$\unlockreq(\Okey, \auth)$$
This message contains the object's key $\Okey$ to unlock (accessible as \\ $\unlockreq.\Okey$) and an authenticator $\auth$ ensuring the user is authorized to unlock $\Okey$.
The authenticator is composed of two parts: (i) a transaction that mutates the object in question and potentially additional objects, which is signed by the object owner, and (ii) a proof that the party requesting an unlocking can modify at least one of the objects in the transaction. The authenticator prevents rogue unlock requests for objects that are either not under contention (the transaction shows there exists a transaction that uses the object) or by parties not authorized to act on the objects.
The user broadcasts this $\unlockreq$ message to all validators~(\one).

Each validator handles the $\unlockreq$ as follows (\Cref{alg:process-unlock-tx}).
A validator first performs the following check:
\begin{itemize}
    \item \textbf{Check (\ref{alg:process-unlock-tx}.1) } It ensures the validity of $\unlockreq$ by verifying the authenticator $\auth$ with respect to the $\Okey$ to unlock. Specifically, it should contain a valid transaction including $\Okey$ and evidence that the unlock is authorized given the owner of $\Okey$. Otherwise stop  processing.
\end{itemize}
The validator attempts to retrieve a certificate $\Cert$ for a transaction on $\Okey$ exists (\textbf{Step (\ref{alg:process-unlock-tx}.2) }) or sets $\Cert$ to \none.
Then, the validator records that the object in $\unlockreq$ can only be included in transaction in the consensus path (\Cref{alg:line:block-exec}) by setting its entry in the $\unlockdb[\Okey]$ to $\unlocked$ (\textbf{Step (\ref{alg:process-unlock-tx}.3)}). It finally returns a signed \emph{unlock vote} $\unlockvote$ to the user:
$$\unlockvote(\unlockreq, \mathbf{Option}(\Cert))$$
This message contains the $\unlockreq$ itself, the (possibly \none) certificate $\cert$ leading to the execution of the object key referenced by $\unlockreq$~(\two).

\noindent
\scalebox{0.75}{
    \noindent
    \begin{minipage}[t]{0.62\textwidth}
        \begin{algorithm}[H]
    \caption{Process unlock requests}
    \label{alg:process-unlock-tx}
    \footnotesize
    \begin{algorithmic}[1]
        \Statex // Handle $\unlockreq$ messages from clients.
        \Procedure{ProcessUnlockTx}{\unlockreq}
        \State // Check (\ref{alg:process-unlock-tx}.1): Check Auth. (\Cref{sec:unlock-protocol}).
        \If{$!\valid{\unlockreq}$} \Return error \EndIf \label{alg:line:verify-auth}
        \State
        \State // Step (\ref{alg:process-unlock-tx}.2): No conflicting executions.
        \State $\Okey \gets \unlockreq.\Okey$
        \State $\Cert \gets \lockdb[\Okey]$ \Comment{Can be $\none$} \label{alg:line:check-cert}
        \State
        \State // Step (\ref{alg:process-unlock-tx}.3): Record the decision to unlock.
        \State $\unlockvote \gets \sign{\unlockreq,\Cert}$
        \State $\unlockdb[\Okey] \gets \unlocked$ \label{alg:line:block-exec}
        \State
        \State \Return $\unlockvote$
        \EndProcedure
    \end{algorithmic}
\end{algorithm}
    \end{minipage}
    \hskip 0.5em
    \begin{minipage}[t]{0.68\textwidth}
        \begin{algorithm}[H]
    \caption{Process unlock certificates}
    \label{alg:process-unlock-cert}
    \footnotesize
    \begin{algorithmic}[1]
        \Statex // Handle $\unlockcert$ message from consensus.
        \Procedure{ProcessUnlockCert}{$\unlockcert$}
        \State // Check (\ref{alg:process-unlock-cert}.1): Check if we can execute (\Cref{sec:unlock-protocol}).
        \If{$\unlockdb[\Okey] = \confirmed$}  \label{alg:line:ensure-first}
        \State \Return
        \EndIf
        \State
        \State // Check (\ref{alg:process-unlock-cert}.2): Check cert validity (\Cref{sec:unlock-protocol}).
        \If{$!\valid{\unlockcert}$} \Return error \EndIf \label{alg:line:check-unlock-cert}
        \State
        \State // Execute $\cert$ or $\none$ (\ref{alg:process-unlock-cert}.3)
        \State $\Cert \gets \unlockcert.\cert$
        \If {$\Cert \neq \none$ }
        \State $\tx \gets \Cert.\tx$
        \Else
        \State $\tx \gets \noop$
        \EndIf
        \State $\effectvote \gets \exec{\tx, \unlockcert}$ \label{alg:line:exec}
        \State
        \State // Prevent execution overwrite.
        \State $\unlockdb[\Okey] \gets \confirmed$ \label{alg:line:mark-confirmed}
        \State
        \State \Return $\effectvote$
        \EndProcedure
    \end{algorithmic}
\end{algorithm}
    \end{minipage}
}
\vskip 1em


The user collects a quorum of $2f+1$ $\unlockvote$ over the same $(\unlockreq, \Cert)$ fields and assembles them into an \emph{unlock certificate} $\unlockcert$:
$$\unlockcert(\unlockreq, \mathbf{Option}(\Cert))$$
where $\unlockreq$ is the certified abort message created by the user and $\Cert$ is the (possibly \none) certificate leading to the execution of the objects referenced by $\unlockreq$.
There are two cases leading to the creation of $\unlockcert$:
\begin{enumerate}
    \item At least one $\unlockvote$ carries a certificate. This scenario indicates that a correct validator already executed a transaction, which implies the object is not locked. However this is not a proof of finality and subsequent steps may invalidate this execution.
    \item No $\unlockvote$ carries a certificate. This scenario is a `no-commit' proof as there are $f+1$ honest validators that will not process certificates ($\unlockdb$ holds $\unlocked$) thus no certificate execution in the fast path will ever become final.
\end{enumerate}
The user submits this $\unlockcert$ for sequencing by the consensus engine~(\three).
All correct validators observe a consistent sequence of  $\unlockcert$ messages output by consensus~(\four) and process them in order as follows (\Cref{alg:process-unlock-cert}).
A validator performs the following checks, and if any fails they ignore the certificate:
\begin{itemize}
    \item \textbf{Check (\ref{alg:process-unlock-cert}.1) } They ensure they did not already process another transaction to completion (i.e.\ $\unlockdb$ is not \confirmed) or a different $\unlockcert$ for the same objects keys.  
    \item \textbf{Check (\ref{alg:process-unlock-cert}.2) } They check $\unlockcert$ is valid, that is, the validator ensures (i) it is correctly signed by a quorum of authorities, and (ii) that the certificate $\cert$ it contains is valid or $\none$.
\end{itemize}
The validator then executes the transaction referenced by $\cert$ (step \ref{alg:process-unlock-cert}.3) if one exists. Otherwise, if $\cert$ is empty, the validator undoes any local transaction executed on the object%
\footnote{The $\unlockcert$ with $\cert$ being \none ensures such an execution could not have been final; only a single layer of execution can ever be undone, and no cascading aborts can happen.}%
, then executes a no-op, that is, the object contents remain unchanged but its version number increases by one.
The validator finally marks every object key as $\confirmed$ to prevent future unlock certificates or checkpoint certificates from overwriting execution (\Cref{alg:line:mark-confirmed}) and returns an $\effectvote$ to the user~(\five).
The user assembles a quorum of $2f+1$ $\effectvote$ messages into an \emph{effect certificate} $\effectcert$ that determines finality~(\six).


Appendix~\ref{sec:gas} details the use of gas objects within the context of \fun and \ref{sec:proofs} proves the safety and liveness of the protocol. The key insight is that an $\unlockcert$ forces transactions on the owned object to go through consensus sequencing. There, either a transaction certificate or an unlock certificate will be sequence first and consistently executed. An unlock certificate for a finalized transaction will always result in the execution of the same transaction.

\para{Auto-Unlock} The basic \fun scheme presumes that the request to unlock an object is authenticated by the owner(s) of the object. This ensures that only authorized parties can interfere with the completion of a transaction, but it also restricts who can initiate unlocking in case of loss of liveness. Alternatively, an `Auto Unlock' scheme may use a synchrony assumption instead to initiate unlock: each validator upon signing a transaction associates with each input object the current timestamp. An Auto Unlock request is identical to a \fun request, but is not authenticated by the object owner. Instead, its validity is checked (checks (\ref{alg:process-unlock-tx}.1) and (\ref{alg:process-unlock-cert}.1)) by ensuring that a sufficient delay $\Delta$ has passed since the object was locked. To ensure liveness the delay $\Delta$ should be long enough to allow for the creation of transaction certificates if there is no contention. \fun and Auto Unlock can be combined: an authenticated request can be processed immediately, but an unauthenticated request is only valid after $\Delta$.

\section{Contention Mitigation} \label{sec:multi-unlock}
The basic \fun protocol speeds up recovery from loss of liveness due to mistakes. However, \sys aims to support workloads on the fast path that are truly under contention. In this case, the basic protocol in \Cref{sec:unlock} is insufficient, since it can result in multiple rounds of locking and no-op unlocking without any user transaction being committed. We present a protocol that proposes a new transaction during the unlock phase that is executed once the unlock is sequenced, ensuring liveness.
%
Additionally, we show how to generalize the basic protocol to unlock multiple objects at once.

\noindent
\scalebox{0.75}{
    \noindent
    \begin{minipage}[t]{0.62\textwidth}
        \begin{algorithm}[H]
    \caption{Process unlock requests (multi)}
    \label{alg:process-unlock-tx-multi}
    \footnotesize

    \begin{algorithmic}[1]
        \Statex // Handle $\unlockreq$ messages from clients.
        \Procedure{ProcessUnlocktTx}{$\unlockreq$}
        \State // Check (\ref{alg:process-unlock-tx-multi}.1): Check authenticator.
        \If{$!\valid{\unlockreq}$} \Return error \EndIf \label{alg:line:check-abort-multi}

        \State
        \State // Collect certificates.
        \State $c \gets \none$
        \For{$\Okey \in \unlockreq.\Okeys$}
        \State $c \gets c \cup \lockdb[\Okey]$ \label{alg:line:check-cert-multi}
        \EndFor
        \State $\unlockvote \gets \sign{\unlockreq, c}$

        \State
        \State // Record the decision to unlock.
        \If {$c==\none$}
        \For{$\Okey \in \unlockreq.\Okeys$}
        \State $\unlockdb[\Okey] \gets \unlocked$ \label{alg:line:block-exec-multi}
        \EndFor
        \State
        \EndIf
        \State \Return $\unlockvote$
        \EndProcedure
    \end{algorithmic}
\end{algorithm}
    \end{minipage}
    \hskip 0.5em
    \begin{minipage}[t]{0.68\textwidth}
        \begin{algorithm}[H]
    \caption{Process unlock certificates (multi)}
    \label{alg:process-unlock-cert-multi}
    \footnotesize
    \begin{algorithmic}[1]
        \Statex // Handle $\unlockcert$ messages from consensus.
        \Procedure{ProcessUnlockCert}{$\unlockcert$}
        \State // Check (\ref{alg:process-unlock-cert-multi}.1): Check message validity.
        \For{$\Okey \in \unlockcert.\Okeys$}
        \If{$\unlockdb[\Okey] = \confirmed$}  \label{alg:line:ensure-first-multi}
        \State \Return
        \EndIf
        \EndFor

        \State
        \State // Check (\ref{alg:process-unlock-cert-multi}.2): Check message validity.
        \If{$!\valid{\unlockcert}$} \Return error \EndIf \label{alg:line:check-abort-multi-cert}

        \State
        \State // Check (\ref{alg:process-unlock-cert-multi}.3): Can we execute the tx?
        \State $v \gets [\;]$
        \If{\unlockcert.\cert = [\;]}  \label{alg:line:no-overwrite-multi}
        \State $\tx \gets \unlockcert.\unlockreq.\tx$
        \State $\effectvote \gets \exec{\tx, \unlockcert}$ \label{alg:line:exec-tx-multi}
        \State $v \gets \effectvote$
        \For{$\Okey \in \unlockcert.\Okeys$}
        \State $\unlockdb[\Okey] = \confirmed$
        \EndFor
        \Else
        \For{$\Cert \in \unlockcert.\cert$}
        \State $\effectvote \gets \exec{\Cert}$ \label{alg:line:exec-multi}
        \State $v \gets v \cup \effectvote$
        \For{$\Okey \in \cert.\Okeys$}
        \State $\unlockdb[\Okey] = \confirmed$ \label{alg:line:mark-confirmed-multi} 
        \EndFor
        \EndFor
        \EndIf

        \State \Return $v$
        \EndProcedure
    \end{algorithmic}
\end{algorithm}
    \end{minipage}
}
\vskip 1em

The multi-objects unlock protocol follows the same general flow as the single-object unlock protocol described in \Cref{sec:unlock}. We thus describe the protocol referring to the steps \one-\six depicted in \Cref{fig:fast-unlock}.

\para{Protocol description}
The user first creates an \emph{unlock request} specifying a set of objects to unlock:
$$\unlockreq([\Okey], \tx, \auth)$$
This message contains a list of the object's keys $[\Okey]$ to unlock (accessible as $\unlockreq.\Okeys$), a new transaction $\tx$ to execute if the unlock process succeeds, and an authenticator $\auth$ ensuring the sender is authorized to access all objects in $[\Okey]$. The user broadcasts this message to all validators~(\one).

\Cref{alg:process-unlock-tx-multi} describes how each validator handles this unlock request $\unlockreq$. They first perform Check (\ref{alg:process-unlock-tx-multi}.1) \Cref{alg:line:check-abort-multi} to check the authenticator $\auth$ is valid with respect to all objects. This check ensures that the user is authorized to mutate all the objects referenced by $\unlockreq$ and to lock all owned object referenced by $\tx$.
The validator then collects any certificates for the objects referenced by $\unlockreq$ (\Cref{alg:line:check-cert-multi}) and adds them to the response as \Cert. The validator then marks object in $\unlockreq$ as reserved for transaction executed through consensus only (\Cref{alg:line:block-exec-multi}).

The validator finally returns an \emph{unlock vote} $\unlockvote$ to the user:
$$\unlockvote(\unlockreq, [\mathbf{Option}(\Cert)])$$
This message contains the unlock message $\unlockreq$ itself and possibly a set of certificates $[\cert]$ on transactions including the object keys referenced by $\unlockreq$ (possible empty)~(\two). If \Cert is not empty the certified transactions may have been finalized, and should be executed instead of the new transaction.

The user collects a quorum of $2f+1$ $\unlockvote$ over the same $\unlockreq$ message and assembles them into an \emph{unlock certificate} $\unlockcert$:
$$\unlockcert(\unlockreq, \Cert)$$
where $\unlockreq$ is the user-created certified unlock message and $U \Cert$ is the unions of all set of certificates received in $\unlockreq$ responses. The user submits this message to the consensus engine~(\three)
The consensus engine sequences all $\unlockcert$ messages; all correct validators observe the same output sequence~(\four).

\Cref{alg:process-unlock-cert-multi} describes how validators process these $\unlockcert$ messages after they are sequenced by the consensus engine.
The validator first ensures they did not already process another $\unlockcert$ or $\cert$ through checkpoint for the same objects keys (\Cref{alg:line:ensure-first-multi}).
They then check $\unlockcert$ is valid, that is, the validator ensures (i) it is correctly signed by a quorum of authorities, and (ii) that all certificates $[\cert]$ it contains are valid (\Cref{alg:line:check-abort-multi-cert}).
The validator can only execute the transaction $\tx$ specified by the user if $\unlockcert.\cert$ is empty (\Cref{alg:line:no-overwrite-multi}).
%
The validator then marks every object key of $[\Okey]$ as $\confirmed$ to prevent any future unlock requests on the $\Okey$ from overwriting execution with a different transaction (\Cref{alg:line:mark-confirmed-multi})
and returns a set of $\effectvote$ to the user~(\five).

The user assembles an $\effectvote$ from a quorum of $2f+1$ validators into an \emph{effect certificate} $\effectcert$ that determines finality~(\six).

\section{Related and Future work}
The \sysname's fast path is based on Byzantine consistent broadcast~\cite{cachinBook}. Previous works suggested using this weaker primitive to build payment systems~\cite{GuerraouiKMPS19,DBLP:journals/corr/abs-1812-10844,astro,fastpay,zef,brick,sui} or even as an exclusion-based locking mechanism for optimistic state-machine replication~\cite{hendricks2010zzyzx}. Zzyzx specifically uses a two-mode unlock mechanism that checks if all replicas have a matching history over the object and retracts the lock or runs full consensus to find the best state to adopt. Unlike Zzyzx, \sysname provides the machinery to not only abort but also directly execute a new transaction and exploits the idea of shared objects to allow for easy execution when there is true contention. Addtionally, \sysname comes with a full set of proofs.

\Cref{sec:background} extensively discussed Sui~\cite{sui}, the closest systems to \sys. Notably, Sui includes a restricted variant of multi-owner transactions to support sponsored transactions, and a restricted variant of complex authenticators allowing only weighted thresholds of signatures as an authenticator.

Sui additionally, supports a batch execution mechanism called Programmable Transaction Blocks (PTB). In a PTB a user can bundle multiple of their transactions together for execution and allows for a significant increase in operations per second Sui can process. Unfortunately, this is currently only available for a single owner largely due to the risk of deadlocks if one of the bundled operations is under a race condition. With \sysname we envision providing this significant advantage in terms of throughput efficiency for general-purpose workloads as dapp operators will be able to bundle transactions of many users in a single certificate workflow knowing that if something goes wrong, they could invoke \fun and seamlessly regain liveness.

%
Another closely related work is FastPay which implements a payment system using a Byzantine consistent broadcast primitive and a lazy synchronizer to achieve \emph{totality}~\cite{cachinBook}. Zef combines FastPay with the Coconut anonymous credentials scheme~\cite{coconut} to enable confidential and unlinkable payments.
Astro relies on an eager implementation of \emph{Byzantine reliable broadcast}~\cite{cachinBook} to achieve totality without relying on an external synchronizer at the cost of higher communication in the common case. Similarly, ABC~\cite{abc} proposes a relaxed notion of consensus where termination is only guaranteed for honest users.
All these systems lack an integration with a consensus path making them both impractical to run for a long-time (no garbage-collection or reconfiguration) as well as limited functionality (only payments) and usability (client-side bugs result in permanent loss of funds). If integrated, then \sysname would apply directly to allow more use-cases on the low latency consensusless path without the risk of locking assets forever due to race conditions.

\section{Conclusion}

\sys proposes a novel approach to decentralized ledgers that addresses the shortcomings of previous consensus-minimized systems. By realizing that the requirement for consensus in blockchains is driven by contention rather than the number of owners, \sys challenges traditional wisdom and provides an alternative perspective. When objects are not under contention, the use of collective objects and multi-owner transactions, combined with the right authentication mechanism enables \sys to give the majority of the functionality seen in traditional blockchains within two round-trips of communication. To properly deal with deadlock when the objects are under contention, \sys proposes the novel \fun protocol allowing users to quickly regain access to locked assets.
As a result, \sys allows for the consensus-less execution of a broader set of transactions, including asset swaps and multi-sig transactions that were previously thought to need consensus.

\ifdefined\publish
    \section*{Acknowledgment}
This work is supported by Mysten Labs. We thank the Mysten Labs Engineering teams for valuable feedback broadly, and specifically to Xun li and Mark Logan for advising on a design that would best fit the Sui codebase.
\fi

\bibliographystyle{plain}
\bibliography{references}

\appendix
\section{Distributed Systems Primitives} \label{sec:primitives}
This section presents the fundamental definitions of the distributed primitives we use as black boxes.

\para{Consensus} Consensus is the process of agreeing on a value or decision among a group of nodes, each of which has its local input and can communicate with other nodes over a network.
Consensus protocols satisfy the following properties:
\begin{itemize}
    \item{Termination:} Eventually, every honest node decides on a value.
    \item{Agreement:} All honest nodes decide on the same value.
    \item{Validity:} If all nodes have the same input value, then any node that decides must decide on that value.
    \item{Integrity:} Nodes only decide on values proposed by some nodes.
\end{itemize}

\para{State Machine Replication}
A Byzantine fault-tolerant state machine replication protocol commits client transactions into a sequential log akin to a single non-faulty server, and provides the following two guarantees:
\begin{itemize}
    \item Safety: Any two honest replicas that commit a transaction at a log position commit the same transaction.
    \item Liveness: A transaction submitted through an honest replica is eventually committed by all honest replicas at some log position.
\end{itemize}
SMR typically uses a consensus instance per log position. In this work, we use the term SMR and consensus interchangeably and assume a black box construction. Any type of consensus that provides safety under asynchrony is sufficient, whether classic~\cite{castro99practical,gelashvili22jolteon} or DAG-based~\cite{danezis22narwhal, bullshark, keidar21all}.

\para{Reliable Broadcast} Reliable Broadcast is a weaker version of consensus that provides liveness only in the presence of an honest \emph{source} node that drives the protocol to completion. As we will discuss later in the paper this is sufficient for the safety and liveness of transactions that do not experience contention on the objects they operate. Unfortunately in the presence of a faulty or buggy source node objects can lose liveness forever.
A reliable broadcast algorithm should satisfy the following properties:
\begin{itemize}
    \item{Validity:} If an honest node broadcasts a message, then every honest node eventually delivers that message.
    \item{Integrity:} If an honest node delivers a message, then that message was previously broadcast by the source.
    \item{Agreement:} If a honest node delivers a message $m$, then every honest node delivers $m$.
\end{itemize}
\section{Epoch Change} \label{sec:reconfiguration}
Sui divides time into a sequence of \emph{epochs}, each comprising an approximately equal number of checkpoints. At the end of each epoch, validators release all their locks to allow users with equivocated objects to regain access to them in the next epoch\footnote{Transactions are valid only within the epoch they were signed to prevent replays.}. However, this design choice introduces a potential risk: what happens to transactions that were in progress during the epoch change?

To mitigate this risk, Sui implements an incremental epoch change process. As a first step, validators pause transaction processing and focus solely on generating checkpoints. This ensures that validators can eventually terminate once all executed transactions are secure. In a second step, validators submit all the certificates they have executed for checkpointing. Once all the certificates known to the validator have been checkpointed, the final step is for the validator to send an \textit{end-of-epoch} message to be sequenced by the consensus protocol. The epoch change is considered complete when $2f+1$ \textit{end-of-epoch} messages are sequenced.

The safety of transactions in progress right before the start of the epoch change process is guaranteed by the requirement that an honest validator must not send an \textit{end-of-epoch} message until all the certificates it executed are sequenced (either by itself or others inserting them into the consensus protocol). This ensures that if a transaction has been finalized in the fast path, at least one honest validator will either include it in a checkpoint or prevent the epoch from completing. As a result, no finalized transaction is reverted or `forgotten' during the epoch change. The preservation of this property can be extended to apply directly to the unlock transactions.

\section{Handling Gas Objects} \label{sec:gas}
Typical transactions not only mutate objects but also consume a gas object to spend for the computation. If, however, the transaction is locked then this gas is blocked as well. For this reason \sysname requires a fresh gas-object in order for consensus to process the unlock request.
Specifically together with~\Cref{alg:process-unlock-tx}, the parties should provide a fresh gas object for their request. This gas object is checked for validity along with the check in~\Cref{alg:line:verify-auth} and locked for the unlock transaction in~\Cref{alg:line:block-exec}. When the user collects the no-commit proof in the second step of the protocol, the $2f+1$ collected signatures also serve as a certificate for the gas object. The consensus then checks the validity of the certificate and spends it locally before entering~\Cref{alg:process-unlock-cert}.
Then when consensus executes the transaction three scenarios may happen:

\begin{itemize}
    \item The unlock request is valid and includes a certificate. Then the execution happens as usual and both the gas object for the unlock and the gas object for the execution are consumed.
    \item The unlock requests is valid and comes with a no-op. Then the gas object for unlock is consumed. If there was some locked transaction racing the \fun then the accompanying gas object is potentially blocked. The user can then explicitly unlock that gas object by running \fun.
    \item The unlock request is not processed because a checkpoint certificate already executed a transaction. Then the gas object is still consumed without altering the state of the $\Okey$.
\end{itemize}
\section{Security Arguments} \label{sec:proofs}

We argue about the safety and liveness of \sysname.
Intuitively, \sysname does not invalidate the finality guarantees of the normal fast path operations. That is, a client holding an effect certificate can be assured that its transaction will not be reverted.

\begin{theorem}\label{th:client-safety}
    If there exists an effect certificate $\effectcert$ over a transaction $\tx$, the execution of $\tx$ is never reverted.
\end{theorem}
\begin{proof}
    We assume that the execution of \tx is reverted and lead to a contradiction. The transaction cannot be reverted at the end of the epoch as it will contradict the properties we inherit from Sui which \sysname did not modify. Hence, the transaction can only be reverted if there exist an $\unlockcert$ over an $\Okey$  modified by $\tx$. For this to happen there should be an $\unlockcert$ over that $\Okey$ carrying an empty certificate.
    From Check (\ref{alg:process-unlock-tx}.2) of \Cref{alg:process-unlock-tx} a correct validator only provides and $\unlockvote$ with an empty $\cert$ if it has not executed anything for $\Okey$. From our assumption that $\Okey$ did admit a no-op there should be $f+1$ honest validators that did not partake in the generation of the $\effectcert$ of $\tx$ and hence passed the check. Additionally, for the $\effectcert$ to exist by definition it has $2f+1$ signatories over the $\Okey$ in question, at least $f+1$ of them being honest.  This implies a total of at least $f+1 + f+1 + f = 3f+2 > 3f+1$ validators, hence a contradiction.
\end{proof}

The dual also applies meaning that if an $\unlockcert$ exists then no $\effectcert$ over the $\Okey$ will be generated in the fast path. The proof analogue by add an extra check during the $\effectcert$ generation that correct validators refuse to process certificates when the recorded $\unlocked$ in their $\unlockdb$.

Next we show that validators that might process on the consensus path both a $\cert$ (through checkpointing) and $\unlockcert$ will arrive at the same execution result. We prove the case where an $\unlockcert$ is ordered first. For this, we need to enhance the protocol of checkpointing in Sui to check the value of $\unlockdb[\Okey]$ and ignore a $cert$ that tries to process a $\confirmed$ $Okey$, which is a straightforward change.

\begin{theorem}
    If a correct validator executes an $\unlockcert$ certificate over $\Okey$ as sequenced by the SMR engine,
    no correct validator will subsequently execute a conflicting a $\cert$ as sequenced by the SMR engine .
\end{theorem}
\begin{proof}
    The proof directly follows from the safety property of the SMR engine that all validators will process certificates in the same order.
    Hence, upon processing $\unlockcert$, all honest validators mark the execution of $\Okey$ as confirmed by setting $\unlockdb[\Okey] \gets \confirmed$ (\Cref{alg:line:mark-confirmed} of \Cref{alg:process-unlock-cert}). Then, check (\ref{alg:process-unlock-cert}.1) of \Cref{alg:process-unlock-cert} (and its dual added at the checkpoint algorithm) ensures that if any further $\cert$ or $\unlockcert$ with a conflict is given as input to the execution engine it is rejected.
\end{proof}

The dual can be proven in the same manner since we enhance the execution of $\cert$ during the checkpoint process with updating $\unlockdb[\Okey] \gets \confirmed$ after processing. Then all $\unlockcert$ on the $\Okey$ will be rejected at the Check (\ref{alg:process-unlock-cert}.1) of \Cref{alg:process-unlock-cert}.

\para{Liveness argument}
Intuitively, we argue that \sysname---and its composition with normal fast path operations---neither deadlocks nor enables unjustified aborts (which could starve an object from progress).

\begin{lemma}[Unlock Certificate Availability] \label{th:certificate-creation}
    A correct user can obtain an unlock certificate $\unlockcert$ over a valid $\Okey$.
\end{lemma}
\begin{proof}
    A correct validator always signs $\unlockvote$ if it passes the check of \Cref{alg:process-unlock-tx}.
    Well formed $\unlockreq$ always come with a valid authentication path (Check (\ref{alg:process-unlock-tx}.1)), and Check (\ref{alg:process-unlock-tx}.2) always returns an $\unlockvote$.
    As a result, if $\unlockreq$ is disseminated to $2f+1$ correct validators by a correct user, they will eventually all return an $\unlockvote$. The user then aggregates those votes into a unlock certificate $\unlockcert$ over $\Okey$.
\end{proof}

\begin{theorem}[\sys Liveness] \label{th:simple-unlock-liveness}
    If a correct and authorized user initiates a fast-unlock protocol, the \Okey in question will eventually admit a new transaction.
\end{theorem}

\begin{proof}
    A correct and authorized user will eventually generate an unlock certificate by \Cref{th:certificate-creation}. Additionally from the liveness property of SMR the unlock certificate will either eventually be added as part of the SMR output or the epoch will end. If the first happens by agreement of consensus the \unlockcert will be executed by all validators, leading to the termination of the fast-unlock protocol and an updated \Okey. If the epoch ends, all locks are dropped and liveness of all \Okey are automatically available for processing.
\end{proof}
\Cref{th:simple-unlock-liveness} is sufficient for correct users as either they will manage to no-op an incorrect invocation of \Okey, drive the tranasction of a correct $\tx$ to completion, or the epoch end will automatically unblock them. This means that there will always be an available $\Okey$ to be modified.

Now that we proved that an authorized user will succeed into unblocking the $\Okey$ we also need to show that an unauthorized user will not succeed into starving legitimate users from progress through abusing fast-unlock.

\begin{theorem}[Starvation Freedom] \label{th:starvation-liveness}
    An user cannot successfully initiate a fast-unlock on any $\Okey$ it is cannot produce an $\auth$.
\end{theorem}

\begin{proof}
    All honest validators check the authorization vector \auth of the requesting user (\Cref{alg:line:verify-auth} in \Cref{alg:process-unlock-tx}). This means that no honest party will lock an object without an authorization, including slow parties that have not yet seen the $\Okey$ which will reject or cache the request for later processing. As a result, by the model, there will never be sufficient $\unlockvote$ to generate an $\unlockcert$ driven by an unauthorized user.
\end{proof}

\para{Generalization to multi-object unlock}
The multi-object unlock protocol can be seen as a composition of many single-object unlock protocols (one per object) as well as a single commit protocol (for the accompanied transaction). As a result, the safety of the protocol follows from the fact that objects are independent of each other so if at least one has a prior certificate then the commit flow will lead to committing that prior certificate (which iteratively applies to all objects with prior certificates). If on the other hand, no object has a prior certificate then the workflow is the combination of the simple \fun per object together with the shared-object path of committing the transactions of Sui which is safe as proven in the original Sui paper~\cite{suiL}.

Second, we explore liveness. There are two cases: (1) all objects can be unlocked, (2) one or more objects are already certified.
The first case is exactly the same as the simple protocol of Section~\ref{sec:unlock} and a proof would follow exactly the same structure. For the second case, we first look into the base case of a single object that is already certified which is already proven in the previous sections. For more than one objects we can see that since the validator adds all certificates in their reply and then processes each certificate separately when handling the unlock cert then there is no interaction between certificate processing and can be considered a batch of independent requests.

Finally, for liveness the accompanied transaction might need to acquire locks. This is also an independent invocation of the Sui fast-path. As a result if the transaction is valid it will either succeed or  blocks. In the latter case, the user will have to invoke fast-unlock again including in the set of to-unlock objects the newly blocked objects of the transaction. Given that there is a finite number of objects a user holds an unlock request will eventually succeed. 

\end{document}